\documentclass[sigconf]{acmart}

\usepackage{algorithm, algorithmicx, algpseudocode}

\AtBeginDocument{%
  \providecommand\BibTeX{{%
    \normalfont B\kern-0.5em{\scshape i\kern-0.25em b}\kern-0.8em\TeX}}}

\setcopyright{acmcopyright}
\copyrightyear{2022}
\acmYear{2022}
\acmDOI{XXXXXXX.XXXXXXX}

\acmConference[CIKM]{International Conference on Information and Knowledge Management}{October 17--22, 2022}{Atlanta, GA}
\acmPrice{15.00}
\acmISBN{978-1-4503-XXXX-X/18/06}

\newtheorem{Theorem}{Theorem}

\newtheorem{Lemma}{Lemma}
\newtheorem{Example}{Example}

\newtheorem{Problem}{Problem}

\newtheorem{Remark}{Remark}

\newtheorem{Assumption}{Assumption}

\newcommand{\sign}{\textup{sign}}

\newcommand{\bequ}{\begin{eqnarray}}
\newcommand{\eequ}{\end{eqnarray}}

\DeclareRobustCommand{\rchi}{{\mathpalette\irchi\relax}}

\newcommand{\irchi}[2]{\raisebox{\depth}{$#1\chi$}}

\begin{document}

\title[FxT Distributed optimization]{Accelerating Distributed Optimization via Fixed-time Convergent Flows: Extensions to Non-convex Functions and Consistent Discretization}


\author{Kunal Garg}
\affiliation{%
 \institution{University of Calfornia, Santa Cruz}
 \streetaddress{1156 High Street}
 \city{Santa Cruz}
 \state{California}
 \country{USA}}

\author{Mayank Baranwal}
\affiliation{%
  \institution{TCS}
  \streetaddress{Address}
  \city{Mumbai}
  \state{Maharashtra}
  \country{India}}

\renewcommand{\shortauthors}{Garg and Baranwal}
\begin{abstract}
Distributed optimization has gained significant attention in recent years, primarily fueled by the availability of a large amount of data and privacy-preserving requirements. This paper presents a fixed-time convergent optimization algorithm for solving a potentially non-convex optimization problem using a first-order multi-agent system. Each agent in the network can access only its private objective function, while local information exchange is permitted between the neighbors. The proposed optimization algorithm combines a fixed-time convergent distributed parameter estimation scheme with a fixed-time distributed consensus scheme as its solution methodology. The results are presented under the assumption that the team objective function is strongly convex, as opposed to the common assumptions in the literature requiring each of the local objective functions to be strongly convex. The results extend to the class of possibly non-convex team objective functions satisfying only the Polyak-\L ojasiewicz (PL) inequality. It is also shown that the proposed continuous-time scheme, when discretized using Euler's method, leads to consistent discretization, i.e., the fixed-time convergence behavior is preserved under discretization. Numerical examples comprising large-scale distributed linear regression and training of neural networks corroborate our theoretical analysis.
\end{abstract}

\keywords{Distributed optimization; Fixed-time stability; Consistent discretization; Multi-agent systems}


\maketitle

\section{Introduction}
Over the past decade, distributed optimization problems over a peer-to-peer network have received considerable attention due to the size and complexity of the dataset, privacy concerns, and communication constraints among multiple agents \cite{nedic2015distributed,lin2017distributed,pan2018distributed}. These distributed convex optimization problems take the following form:

\begin{align}\label{eq:prob}
    \min\limits_{\mathbf{x}\in\mathbb{R}^d} \ \ F(\mathbf{x}) = \sum\limits_{i=1}^Nf_i(\mathbf{x}),
\end{align}
where $F(\cdot)$ is the team objective function, and the convex function $f_i: \mathbb{R}^d\to\mathbb{R}$ represents the local objective function of the $i^\text{th}$ agent, where $i \in\{1,2, \ldots, N\}$ for some positive integer $N$. Distributed optimization problems find applications in several domains including, but not limited to, sensor networks~ \cite{rabbat2004distributed}, satellite tracking~\cite{hu2016smooth}, and large-scale machine learning~\cite{nathan2017optimization}. Distributed optimization problems facilitate distributed coordination among the agents, as well as minimization of the team objective function. Consequently, these problems are inherently more complex than other multi-agent control problems, such as, distributed consensus.

In recent years, the use of continuous-time dynamical systems for distributed optimization has emerged as a viable alternative \cite{lin2017distributed,pan2018distributed,feng2017finite,hu2018distributed,su2014differential}. This viewpoint enables the use of tools from Lyapunov theory and differential equations for the analysis and design of optimization procedures. It is worth mentioning that most of the existing continuous-time schemes for distributed optimization are only asymptotically (or exponentially at best) convergent. On the other hand, most practical multi-agent optimization tasks, such as distributed economic dispatch, often undergo frequent changes in operating conditions, thereby requiring the optima to be achieved in a finite amount of time.

The notion of finite-time convergence in optimization is closely related to finite-time stability \cite{bhat2000finite} in control theory. In contrast to asymptotic stability (AS), finite-time stability is a concept that guarantees the convergence of solutions in a finite amount of time. In~\cite{lu2012zero}, a continuous-time zero-gradient-sum (ZGS) with an exponential convergence rate was proposed, which, when combined with a finite-time consensus protocol, was shown to achieve finite-time convergence in~\cite{feng2017finite}. A drawback of ZGS-type algorithms is the requirement of strong convexity of the local objective functions and the choice of specific initial conditions $x_i(0)$ for each agent $i$ such that $\sum_{i=1}^N\nabla f_i(x_i(0))=0$. In \cite{lin2017distributed}, a novel continuous-time distributed optimization algorithm, based on private (nonuniform) gradient gains, was proposed for convex functions with quadratic growth and achieved convergence in a finite time. A finite-time tracking and consensus-based algorithm were recently proposed in~\cite{hu2018distributed}, which again achieves convergence in a finite time under a time-invariant communication topology.

Fixed-time stability (FxTS) \cite{polyakov2012nonlinear} is a stronger notion than finite-time stability (FTS), where the time of convergence does not depend upon the initial condition. To the best of our knowledge, distributed optimization procedures with fixed-time convergence have not been addressed in the literature for a general class of non-linear, potentially non-convex, objective functions. The use of FxTS theory for distributed optimization was first investigated in~\cite{garg2020fixed} where centralized optimization problems were studied. The authors in~\cite{wang2020distributed} further specialized it to the case of strongly convex functions, however, at the expense of using a Hessian-based (second-order) schemes that do not scale well with the dimension $d$ of the underlying state-space. Moreover, the distributed protocol in~\cite{wang2020distributed} requires each of the individual private objective functions to be strongly convex. In the particular case of quadratic objective functions, the scheme proposed in~\cite{garg2020fixed} can be suitably modified to incorporate both inequality and equality constraints~\cite{baranwal2020robust}.

Despite growing interests in the use of continuous-time dynamical systems towards distributed optimization with fixed-time convergence guarantees, the existing literature makes various simplifying assumptions, including but not limited to, requiring agents to satisfy ZGS condition, use of second-order (Hessian-based) optimization schemes, necessitating all private objective functions to be strongly convex or with bounded growth, and existence of a time-invariant communication topology. Most of these requirements limit the power of fixed-time convergent dynamical systems towards being adopted for practical cooperative multi-agent control problems. Finally, prior work does not discuss how efficient their proposed methods are during implementation using iterative, discrete methods. It is worth noting that while continuous-time dynamical systems are studied for ease of understanding the behavior of an optimization algorithm, in practice, it is inevitable to use a discrete-time, iterative method to solve optimization problems. In light of this, inspired from the work in \cite{polyakov2019consistent,benosman2020optimizing} and using the recent results from \cite{garg2021MVIP}, we show that our proposed method leads to a \textit{consistent} discretization scheme where the fixed-time convergent behavior is preserved upon discretization using elementary schemes, such as Euler discretization.

In view of the limitations stated above, this paper presents a fixed-time convergent, distributed optimization scheme for first-order multi-agent systems that extend to a broad class of local objective functions under relaxed assumptions on convexity and information to be exchanged with the neighbors. 
The main contributions of the paper are summarized below:
\begin{itemize}
    \item We consider the problem of distributed optimization of the sum of local objective functions, assuming that only the global objective function is strongly convex. Unlike prior works, we do not require each of the local objective functions to be strongly convex. 
    \item The results are extended to a class of possibly non-convex functions satisfying only the Polyak-\L ojasiewicz (PL) inequality. PL inequality is a relaxation of strong-convexity and is popularly used to design exponentially stable gradient-flows in the centralized optimization problems \cite{hassan2019proximal,garg2018fixed}.
    To the best of the authors' knowledge, this is the first work that utilizes this condition in distributed optimization. 
    \item We show that trajectories of dynamics obtained by discretizing the proposed continuous-time dynamics using Euler discretization converge to an arbitrarily small neighborhood of the optimal point within a fixed number of iterations, leading to a \textit{consistent} discretization. This is a rather significant result as it bridges the gap between the continuous-time analysis and discrete-time implementation and is skipped by almost all of the prior work on the dynamical system-based approach to solving optimization problems.
    \item Finally, we validate the proposed distributed optimization algorithm for decentralized learning of regression parameters in a linear regression task and training deep neural networks for classification on the MNIST dataset.
\end{itemize}

\noindent {\bf A note on mathematical notations}: We use $\mathbb{R}$ to denote the set of real numbers and $\mathbb R_+$ to denote non-negative real numbers. Given a function $f:\mathbb R^d\rightarrow\mathbb R$, the gradient and the Hessian of $f$ at some point $x\in \mathbb R^d$ are denoted by $\nabla f(x)$ and $\nabla^2 f(x)$, respectively. Given $x\in\mathbb{R}^d$, $\|x\|$ denotes the 2-norm of $x$. $\mathcal{G}=(A,\mathcal{V})$ represents an undirected graph with the adjacency matrix $A = [a_{ij}]\in \mathbb R^{N\times N}$, $a_{ij}\in\{0,1\}$ and the set of nodes $\mathcal{V} = \{1, 2, \cdots, N\}$. The set of 1-hop neighbors of node $i\in \mathcal V$ is represented by $\mathcal{N}_i$, i.e., $\mathcal N_i(t) = \{j\in\mathcal{V}\; |\; a_{ij} = 1\}$. The second smallest eigenvalue of a matrix is denoted by $\lambda_2(\cdot)$. We define the function $\sign^\mu: \mathbb R^d\rightarrow\mathbb R^d$ as 
\begin{align}\label{sign func}
    \sign^\mu(x) = x\|x\|^{\mu-1}, \ \ \mu\geq 0,
\end{align}
with $\sign^\mu(0) = 0$. We use $1_N, 0_N\in \mathbb R^N$ to denote vectors consisting of ones and zeros, respectively, of dimension $N$. 

\section{Problem Formulation and Preliminaries}\label{sec:prob_form}
\subsection{Problem statement}
Consider the system consisting of $N$ nodes with graph structure $\mathcal G = (A,\mathcal V)$ specifying the communication links between the nodes for $t\geq 0$. The objective is to find $x^*\in\mathbf{R}^d$ that solves
\begin{align}\label{opt prob}
\begin{split}
    & \min_{x_1, x_2. \cdots, x_N} \sum_{i=1}^{N}f_i(x_i), \\
    \textrm{s.t.}& \; x_1 = x_2 = \cdots = x_N. 
\end{split}
\end{align}
In this work, we assume that the minimizer $x^*=x_1^*=x_2^*=\cdots=x_N^*$ for \eqref{opt prob} exists and is unique.\footnote{Existence and uniqueness of global minimizer is trivially satisfied for a strongly convex team objective function. While the PL inequality (see Assumption \ref{f assum 2}) does not imply convexity, it implies invexity, i.e., the stationary points are global minimizers.
} We make the following assumption on the inter-node communications.
\begin{Assumption}\label{ass:1}
	The communication topology between the agents is connected and undirected, i.e., the underlying graph $\mathcal G = (A,\mathcal V)$ is connected, and $A$ is a symmetric matrix.
\end{Assumption}

To motivate the dynamical system approach considered in this paper, first, let us revisit the gradient decent (GD) method to minimize an unconstrained function $\mathcal F:\mathbb R^n\rightarrow\mathbb R$, given as:
\begin{align*}
    x_{k+1} = x_k -\eta~\nabla \mathcal F(x_k),
\end{align*}
where $\eta>0$ is the step-size. We can re-write the above as $\frac{x_{k+1}-x_k}{\eta} = -\nabla \mathcal F(x_k)$ and in the limit $\eta\to 0$, we obtain the continuous-time equivalent of GD, termed as gradient-flow, given as $\dot x = -\nabla \mathcal F(x)$. More generally, we can write this dynamical system as $\dot x = u$ where $u$ can be designed to solve a given problem (e.g., for unconstrained minimization of $\mathcal F$, $u = -\nabla \mathcal F$ and for constrained minimization of $\mathcal F$ over a convex set $\mathcal C$, one can define $u = -k (x-\textrm{P}_{\mathcal C}(x-\nabla \mathcal F(x)))$ using the projection operator $\textrm{P}_{\mathcal C}$. Inspired from this, we use a dynamical system approach to solve the constrained optimization problem \eqref{opt prob} in a distributed fashion. Let $x_i\in\mathbb{R}^d$ represent the state of agent $i$. We model agent $i$ as a first-order integrator system:
\begin{align}\label{sys flow}
	\dot x_i = u_i,
\end{align}
where $u_i\in\mathbb{R}^d$ can be regarded as a \textit{control input}, that depends upon the states of the agent $i$, and the states of the neighboring agents $j_1, j_2, \cdots, j_l\in \mathcal N_i$.
The problem statement is formally given as follows.\vspace{-.5em}
\begin{Problem}
Design 
$u_i$ for each agent $i\in \mathcal V$, such that $x_1=x_2=\cdots=x_N=x^*$ is achieved under \eqref{sys flow} within a fixed fixed time, for any initial condition $\{x_1(0), x_2(0), \cdots, x_N(0)\}$, where $x^*$ solves~\eqref{opt prob}.
\end{Problem}

\subsection{Preliminaries}
In this subsection, we present relevant definitions and results on FxTS. Consider the system: 
\begin{align}\label{ex sys}
	\dot x = \phi(x),
\end{align}
where $x\in \mathbb R^d$, $\phi: \mathbb R^d \rightarrow \mathbb R^d$ and $\phi(0)=0$. The authors in \cite{polyakov2012nonlinear} presented the following result for fixed-time stability, where the time of convergence is finite and is uniformly bounded for any initial condition $x(0)$.
\vspace{-.5em}
\begin{Lemma}[\cite{polyakov2012nonlinear}]\label{FxTS TH}
Suppose there exists a positive definite, radially unbounded, continuously differentiable function $V:\mathbb R^d\rightarrow\mathbb R$, i.e., $V\in \mathcal C^1$ such that $V(0) = 0$ and $V(x)>0$ for $x\neq 0$, such that the following holds:
\begin{align}\label{eq:polyakov}
    \dot V(x) \leq -aV(x)^p-bV(x)^q, \quad \forall x \neq 0,
\end{align}
with $a,b>0$, $0<p<1$ and $q>1$. Then the origin of \eqref{ex sys} is FxTS, i.e., $x(t) = 0$ for all $t\geq T$, where the settling time $T$ satisfies $T \leq \frac{1}{a(1-p)} + \frac{1}{b(q-1)}$. 
\end{Lemma}

Next, we present some well-known results that will be useful in proving our claims on fixed-time parameter estimation and consensus protocols. 
\begin{Lemma}[\cite{zuo2016distributed}]\label{lemma:ineq}
	Let $z_i\in \mathbb R_+$ for $i\in \{1,2,\cdots, N\}$, $N\in \mathbb Z_+$. Then the following hold:
	\begin{subequations}\label{p ineq}
		\begin{align}
		\sum_{i = 1}^N z_i^p& \geq \left(\sum_{i = 1}^Nz_i\right)^p, \; 0<p\leq 1,\\
		\sum_{i = 1}^N z_i^p& \geq N^{1-p}\left(\sum_{i = 1}^Nz_i\right)^p, \; p>1.
		\end{align}
	\end{subequations}
\end{Lemma}

\begin{Lemma}\label{lemma:sign}
	Let $\mathcal G = (A,\mathcal V)$ be an undirected graph consisting of $N$ nodes located at $x_i\in \mathbb R^d$ for $i\in \{1,2,\cdots, N\}$ and $\mathcal N_i$ denotes the in-neighbors of node $i$. Then,  
	\begin{align}\label{sign 0}
	\sum_{i = 1}^N\sum_{j \in \mathcal N_i}\sign(x_i-x_j) = 0.
	\end{align}
\end{Lemma}

\begin{Lemma}\label{lemma eij half}
	Let $w:\mathbb R^d\rightarrow \mathbb R^d$ be an odd mapping, i.e., $w(x) = -w(-x)$ for all $x\in \mathbb R^d$ and let the graph $\mathcal G = (A,\mathcal V)$ be undirected. Let $\{x_i\}$ and $\{e_i\}$ be the sets of arbitrary vectors with $i\in \mathcal V$ and $x_{ij} \coloneqq x_i-x_j$ and $e_{ij} \coloneqq e_i-e_j$. Then, the following holds
	\begin{align}\label{f ij  eij}
	    \sum_{i,j= 1}^Na_{ij}e_i^{\intercal}w(x_{ij}) = \frac{1}{2}\!\sum_{i,j= 1}^Na_{ij}e_{ij}^{\intercal}w(x_{ij}).
	\end{align}
\end{Lemma}

\begin{Lemma}[\cite{mesbahi2010graph}]\label{lemma:Laplacian}
	Let $\mathcal G = (A,\mathcal V)$ be an undirected, connected graph. Let $L_A\coloneqq[l_{ij}]\in\mathbb{R}^{N\times N}$ be graph Laplacian matrix defined as $
	l_{ij} = \left\{
	\begin{array}{cc}
    	\sum\limits_{k=1,k\neq i}^{N}a_{ik}, & i=j\\
    	-a_{ij}, & i\neq j
	\end{array}
	\right.$. Then the Laplacian $L_A$ has following properties:\\
	1) $L_A$ is positive semi-definite, $L_A 1_N ={0}_N$, and $\lambda_2(L_A) > 0$.\\
	2) $x^{\intercal}L_Ax=\frac{1}{2}\sum_{i,j=1}^Na_{ij}(x_j-x_i)^2$, and if $1^{\intercal}x=0$, then $x^{\intercal} L_Ax\geq \lambda_2(L_A)x^{\intercal} x$.
\end{Lemma}

\section{Main results}\label{sec:strong_convex}
Our approach to fixed-time multi-agent distributed optimization is based on first designing a centralized fixed-time protocol that relies upon global information. Then, the quantities in the centralized protocol are estimated in a distributed manner. In summary, the algorithm proceeds by first estimating global quantities ($g^*$ as defined in \eqref{eq:g_star new}) required for the centralized protocol, then driving the agents to reach consensus ($x_i(t) = x(t)$ for all $i\in \mathcal V$), and finally driving the \textit{common} trajectory $x(t)$ to the optimal point $x^*$, all within a fixed time $T$. Recall that agents are said to have reached consensus on states $x_i$ if $x_i = x_j$ for all $i, j\in \mathcal V$. To this end, we define first a centralized fixed-time protocol. Note that agents' states are driven by the same input under centralized settings and are initialized to the same starting point. In a distributed setting, this behavior translates to agents having already reached consensus and subsequently being driven by a common input (see Remark \ref{remark cent dist}). This section presents a Hessian-free, first-order dynamical system that achieves convergence to the global optimum of strongly convex team objective function in a fixed time. We make the following assumptions for the results in this section. 

\begin{Assumption}\label{ass:2 new}
	Functions $f_i$ are convex, twice differentiable and the Hessian $\nabla^2F(x)=\sum_{i = 1}^N\nabla^2f_i(x) \succeq kI$, where $k>0$, for all $x\in \mathbb R^d$, i.e., function $F$ is strongly convex with modulus $k$.
	\vspace{-.5em}
\end{Assumption}
\begin{Remark}
    Assumption \ref{ass:2 new} can be satisfied even if just one of the objective functions is strongly convex. Assumption \ref{ass:2 new} also implies that $x = x^*$ is a minimizer if and only if it satisfies $\sum_{i=1}^N\nabla f_i(x^*) = 0$.
    \vspace{-.5em}
\end{Remark}
\begin{Assumption}\label{ass:3 new}
    Each node $i$ receives $x_j, \nabla f_j(x_j)$ from each of its neighboring nodes $j\in \mathcal N_i$.
    \vspace{-.5em}
\end{Assumption}

Note that under Assumption~\ref{ass:2 new}, the agents only need to exchange their state values $x_i$ and the gradients $\nabla f_i(x_i)$ with their neighbors, and there is no need to exchange the Hessian values under this framework. We first present a centralized protocol that guarantees solution of \eqref{opt prob} in a fixed time. All the results in the following section assume that Assumption~\ref{ass:1}, \ref{ass:2 new}, \ref{ass:3 new} hold, unless specified otherwise. 

\subsection{Centralized protocol}

\begin{Lemma}[\textbf{Centralized fixed-time protocol}]\label{thm:central new}
Suppose the dynamics of each agent $i\in\mathcal{V}$ in the network is given by
\begin{align}\label{eq:cent_dynamics}
	u_i = g^*, \ \ \ \ x_i(0) = x_j(0) \ \ \ \ \forall \  i,j\in\mathcal{V},
\end{align}
where $g^*$ is defined as:
\begin{align}\label{eq:g_star new}
    g^*(x) &= -\left(\sum_{i = 1}^N \nabla f_i(x) + \sign^{l_1}\left(\sum_{i = 1}^N\nabla f_i(x)\right)\right.\nonumber\\
    &\qquad \left.+\sign^{l_2}\left(\sum_{i = 1}^N\nabla f_i(x)\right) \right)
\end{align}
where $l_1>1$ and $0<l_2<1$, and $x_i = x$ for each $i\in\mathcal{V}$, for all $t\geq 0$. Then the trajectories of all agents converge to the optimal point $x^*$, i.e., the minimizer of the team objective function \eqref{opt prob} in a fixed time $T_{sc}>0$.
\end{Lemma}

\begin{proof}
Consider a candidate Lyapunov function: $$V(x)=\frac{1}{2}\left(\sum_{i = 1}^N\nabla f_i(x)\right)^{\intercal}\left(\sum_{i = 1}^N\nabla f_i(x)\right).$$ 
By taking its time-derivative along \eqref{eq:cent_dynamics}, we obtain:

\begin{align*}
	\dot V(x) &= \left(\sum_{i = 1}^N\nabla f_i(x)\right)^{\intercal}\left(\sum_{i = 1}^N\nabla^2 f_i(x)\dot{x}\right) \\
	& = -\left(\sum_{i = 1}^N\nabla f_i\right)^{\intercal}\left(\sum_{i = 1}^N\nabla^2 f_i\right)\left(\sum_{i = 1}^N \nabla f_i \right.\\
	&\qquad \left.+\sign^{l_1}\left(\sum_{i = 1}^N\nabla f_i\right) +\sign^{l_2}\left(\sum_{i = 1}^N\nabla f_i\right)\right),\\
	& \leq -2kV-k\left\|\sum_{i = 1}^N\nabla f_i\right\|^{l_1+1}-k\left\|\sum_{i = 1}^N\nabla f_i\right\|^{l_2+1}\\
	&\leq -k2^\frac{1+l_1}{2}V^{\frac{1+l_1}{2}}-k2^\frac{1+l_2}{2}V^{\frac{1+l_2}{2}},
\end{align*}
where the first inequality follows from the fact that $\sum_{i = 1}^N\nabla^2f_i\succeq kI$. Thus, using Lemma~\ref{FxTS TH}, we have that there exists $T_{sc}<\infty$ such that for all $t\geq T_{sc}$, $x(t) = x^*$ starting from any initial condition.
\end{proof}

The centralized fixed-time protocol inherently assumes that the agents can directly access the global quantity $\sum_{i=1}^N\nabla f_i$. In a distributed setting, this quantity needs to be estimated and is not directly accessible. Before presenting the algorithm to compute this global quantity in a distributed manner, we first present an extension of Lemma~\ref{thm:central new} under further relaxation of Assumption \ref{ass:2 new}. The notion of \textit{gradient-dominance} or Polyak-\L ojasiewicz (PL) inequality has been explored extensively in optimization literature to show exponential convergence. A function $f:\mathbb R^n\rightarrow\mathbb R$ is said to satisfy PL inequality, or is gradient dominated, with $\mu_f>0$ if
\begin{equation}\label{PL ineq}
     \frac{1}{2}\|\nabla f(x)\|^2\geq \mu_f(f(x)-f^*) \quad \forall x \in \mathbb R^n,
\end{equation}
where $f^* = f(x^*)$ is the value of the function at its minimizer $x^*$. We make the following assumption on the team objective function. 

\begin{Assumption}{(\textbf{Gradient dominated})} \label{f assum 2}
The function $F$ is radially unbounded, has a unique minimizer $x = x^*$, and satisfies the PL inequality, or is gradient dominated, i.e., there exists $\mu>0$ such that
\begin{align}\label{eq: F PL ineq}
     \hspace{-10pt}\frac{1}{2}\left\|\sum_{i = 1}^N\nabla f(x)\right\|^2 \geq \mu(F(x)-F^*) = \mu\sum_{i = 1}^N(f_i(x)-f_i^*),
\end{align}
where $F^* = F(x^*)$ and $f_i^* = f_i(x^*)$.
\vspace{-.5em}
\end{Assumption}

\begin{Remark}\label{PL remark}
Strong convexity of the objective function is a  standard assumption used in literature to show exponential convergence. As noted in \cite{karimi2016linear}, PL inequality is the weakest condition among other similar conditions popularly used in the literature to show linear convergence in discrete-time (exponential, in continuous-time). Notably, a strongly convex function $F$ satisfies PL inequality. Furthermore, note that under Assumption~\ref{f assum 2}, it is not required that the function $F$ is convex, as long as its minimizer exists and is unique.
\end{Remark}

It is easy to show that if a function $F:\mathbb R^m\rightarrow \mathbb R$ is strongly convex, then the function $G:\mathbb R^n\rightarrow \mathbb R$, defined as $G(x) = F(Ax)$, where $A\in \mathbb R^{n\times m}$ is not full row-rank, may not be strongly convex. On the other hand, as shown in \cite[Appendix 2.3]{karimi2016linear}, $G$ still satisfies PL inequality for any matrix $A$. Below, an example of an important class of problems is given for which the objective function satisfies PL inequality.

\begin{Example}\label{Example: PL inequality}
\textbf{Least squares}:  Consider the optimization problem 
\begin{equation}\label{PL example 1}
    \min_x \|Ax-b\|^2 = \sum_{i = 1}^n \|A_ix-b_i\|^2,
\end{equation}
where $x\in \mathbb R^n, A\in \mathbb R^{n\times n}$ and $b\in \mathbb R^n$. Here, the function $F(x) = \|x-b\|^2$ is strongly-convex, and hence, $G(x) = \|Ax-b\|^2$ satisfies PL inequality for any matrix $A$. 
\end{Example}

The objective function of \eqref{PL example 1} satisfies PL inequality, but need not be strongly convex for any matrix $A$, thus, one can use \eqref{eq:cent_dynamics} to find the optimal solution for \eqref{PL example 1} in a fixed time. This is an important class of functions in machine learning problems. 
 
\begin{Lemma}\label{corollary central PL}
Let Assumption \ref{f assum 2} hold. Suppose the dynamics of each agent $i\in\mathcal{V}$ in the network is given by \eqref{eq:cent_dynamics} where $g^*$ given as \eqref{eq:g_star new} with $x_i(t) = x(t)$ for each $i\in\mathcal{V}$, for all $t\geq 0$. Then the trajectories of all agents converge to the optimal point $x^*$, i.e., the minimizer of the team objective function \eqref{opt prob} in a fixed time $T_{PL}>0$.
\end{Lemma}
\begin{proof}
Consider the candidate Lyapunov function as $V(x) = \sum_{i = 1}^N(f_i(x)-f_i(x^*)) = (F(x)-F^*)$. Note that $V$ is positive definite and per Assumption \ref{f assum 2}, radially unbounded. Taking its time derivative along the trajectories of \eqref{eq:cent_dynamics}, we obtain
\begin{align*}
    \dot V(x) & = -\sum_{i = 1}^N\nabla f_i^{\intercal}\left(\sum_{i = 1}^N \nabla f_i + \sign^{l_1}\left(\sum_{i = 1}^N\nabla f_i\right)\right. \\
    &\quad +\left.\sign^{l_2}\left(\sum_{i = 1}^N\nabla f_i\right)\right)\\
    & = -\|\nabla F(x)\|^2-\|\nabla F(x)\|^{l_1+1}-\|\nabla F(x)\|^{l_2+1}\\
    & \overset{\eqref{eq: F PL ineq}}{\leq} -2\mu (F(x)-F^*)-(2\mu)^{\frac{1+l_1}{2}}(F(x)-F^*)^{\frac{1+l_1}{2}}\\
    &\qquad -(2\mu)^{\frac{1+l_2}{2}}(F(x)-F^*)^{\frac{1+l_2}{2}}\\
    & \leq -4\mu V(x)-(4\mu)^{\frac{1+l_1}{2}}V(x)^\frac{1+l_1}{2}-(4\mu)^{\frac{1+l_2}{2}}V(x)^\frac{1+l_2}{2} \\
    & \leq -k_1 V(x)^\frac{1+l_1}{2}-k_2 V(x)^\frac{1+l_2}{2}.
\end{align*}
Thus, using Lemma~\ref{FxTS TH}, we obtain that there exits $T_{PL}<\infty$ such that for all $t\geq T_{PL}$, we have that $V(x(t)) = 0$, or equivalently, $F(x(t)) = F^*$. Under Assumption \ref{f assum 2}, we have that $F$ has a unique minimizer, and thus, $F(x(t)) = F^*$ implies that $x(t) = x^*$, which completes the proof.
\end{proof}

\begin{Remark}\label{remark cent dist}
    Lemmas~\ref{thm:central new} and \ref{corollary central PL} represent centralized protocols for convex optimization of team objective functions. Here, the agents are already in consensus and have access to the global information $\sum_{i = 1}^N \nabla f_i(x)$. In the distributed setting, agents can only access their local information, as well as $x_j$, $\nabla f_j(x_j)$ for all $j\in\mathcal{N}_i$, and will not be in consensus in the beginning. We now propose distributed scheme for estimation of global quantity that achieves consensus in a fixed time. The main distributed algorithm is presented in Algorithm \ref{alg:opt} at the end of this section. 
\end{Remark}

\subsection{Distributed estimation of global parameters}
We now present results for distributed estimation of global quantity that achieves consensus in a fixed time so that the problem can be solved in a distributed setting. 
For each agent $i\in\mathcal{V}$, define $g_i$ as:
\begin{align}\label{eq:gi_func new}
\hspace{-10pt}g_i & = -\left(N\theta_{i} + \sign^{l_1}(N\theta_{i})+ \sign^{l_2}(N\theta_{i})\right),
\end{align}
where $g_i$ denotes agent $i$'s estimate of $g^*$ and $\theta_i:\mathbb R_+\rightarrow\mathbb R^{d}$ is the estimate of the global (centralized) quantities, whose dynamics is defined as
\begin{align}\label{eq:theta new}
	\dot {\theta}_i = \omega_i + h_i,
\end{align}
where $h_i$ is defined as $h_i = \frac{d}{dt}\nabla f_i(x_i)$. The signal $\omega:\mathbb R_+\rightarrow\mathbb R^{d}$, defined as
\begin{align}\label{eq:omega_dot new}
	\omega_i &= p\sum_{j\in \mathcal N_i}\Big(\sign(\theta_j-\theta_i) + \gamma \sign^{\nu_1}(\theta_j-\theta_i) + \delta \sign(\theta_j-\theta_i)^{\nu_2}\Big),
\end{align}
where $p, \gamma, \delta>0$, and $0<\nu_2<1<\nu_1$, are suitably chosen in order to achieve consensus over the quantities $\theta_i$, as shown later. The functions $\{h_i\}$ are needed to drive this average consensus values to the global quantities to be estimated. Observe that $\{\theta_i\}$ are updated in~\eqref{eq:theta new} in a \emph{distributed} manner. We make the following assumption on functions $h_i$.

\begin{Assumption}\label{ass:4}
	The functions $h_i, h_j$ satisfy $\|h_i(t)-h_j(t)\|\leq \rho$ for all $t\geq 0$, $i,j \in \mathcal{V}, i\neq j$, for some $\rho>0$.
	\vspace{-.5em}
\end{Assumption}
This assumption can be easily satisfied if the graph is connected for all time $t$  and the gradients and their derivatives are bounded \cite{hu2018distributed, rahili2016distributed}. Many common objective functions, such as quadratic cost functions satisfy this assumption. Under this assumption, we can state the following results. 

\begin{Lemma}\label{lemma:app}
Let Assumption \ref{ass:4} hold, and the gain $p$ in \eqref{eq:omega_dot new} satisfies $p>\left(\frac{N-1}{2}\right)\rho$; then for each agent $i\in\mathcal{V}$, $\theta_i(t)=\theta_c(t) \coloneqq \frac{1}{N}\sum_{j=1}^N\theta_j(t) = \frac{1}{N}\sum_{i = 1}^N\nabla f_i(x_i(t))$, for all $t\geq T_p$ where 
$$T_\text{p} \coloneqq \frac{2}{p\gamma N^{2(1-\kappa_1)}c^{\kappa_1}(\kappa_1-1)}+\frac{2}{p\delta c^{\kappa_2}(1-\kappa_2)},$$ $\kappa_1 = \frac{1+\nu_1}{2}, \kappa_2= \frac{1+\nu_2}{2}$ and $c = 4\lambda_2(L_A)$.
\end{Lemma}

\noindent The proof is provided in Appendix \ref{app: lemma app proof}. We now present the following result on distributed parameter estimation in a fixed time.

\begin{Lemma}[\textbf{Fixed-time parameter estimation}]\label{thm:param_estimation new}
	Let $\omega_i(0) = \mathbf 0_{d}$ for each $i\in \mathcal V$ and the gain $p$ in \eqref{eq:omega_dot new} satisfy $p>\left(\frac{N-1}{2}\right)\rho$. Then there exists a fixed-time $0<T_\text{p}<\infty$ such that $g_i(t) = g_j(t)$ for all $i,j\in\mathcal{V}$ and $t\geq T_\text{p}$.
\end{Lemma}

\begin{proof}
	The proof follows directly from Lemma~\ref{lemma:app}, i.e., it holds that $\theta_i(t) = \theta_j(t)$ for all $t\geq T_\text{p}$, $i, j\in \mathcal V$. From the definition of $g_i$ in \eqref{eq:gi_func new}, it follows that $g_i(t)=g_j(t)$ for all $t\geq T_\text{p}$ and for each $i,j\in\mathcal{V}$.
\end{proof}

\noindent The centralized fixed-time protocol in Lemma~\ref{thm:central new} is based on two key assumptions: (a) Agents are being driven by the same input $g^*$, and (b) agents start at the same initial state,, i.e., $x_i(0)=x_j(0)$ for all $i,j\in\mathcal{V}$. To this end, Lemma~\ref{thm:param_estimation new} only ensures that the first of the two conditions is met. All agents must be driven to the same state in order to ensure the applicability of Lemma~\ref{thm:central new} in the distributed setting. Consequently, we propose the following update rule for each agent $i\in\mathcal{V}$ in the network:
\begin{align}\label{eq:u_i new}
	u_i = \tilde u_i + g_i,
\end{align}
where $g_i$ is as described in \eqref{eq:gi_func new}, and $\tilde{u}_i$ is defined as locally averaged signed differences:
\begin{align}\label{eq:u_tilde}
	\tilde u_i &= q\sum_{j\in \mathcal N_i}\Big(\sign(x_j-x_i) + \alpha \sign^{\mu_1}(x_j-x_i)  +\beta \sign^{\mu_2}(x_j-x_i)\Big),
\end{align}
where $q, \alpha, \beta >0$, $\mu_1>1$ and $0<\mu_2<1$. The following results establish that the state update rule for each agent proposed in \eqref{eq:u_i new} ensures that the agents reach global consensus and optimality in fixed time.

\begin{Lemma}[\textbf{Fixed-time consensus}]\label{thm:consensus new}
	Under the effect of control law $u_i$ in \eqref{eq:u_i new} with $\tilde{u}_i$ defined as in \eqref{eq:u_tilde}, and $g_i(t)=g_j(t)$ for all $t\geq T_\text{p}$ and $i,j\in\mathcal V$, the closed-loop trajectories of \eqref{sys flow} converge to a common point $x$ for all $i\in \mathcal V$ in a fixed time $T_\text{con}$, i.e., $x_i(t)=\bar{x}(t)$ for all $t\geq T_\text{p}+T_\text{con}$.
\end{Lemma}

\begin{proof}
	The proof follows from Lemma~\ref{lemma:app} and the fact that $g_i(t)=g_j(t)$ for all $t\geq T_\text{p}$, $i,j\in\mathcal{V}$. Thus, for $t\geq T_\text{p}$, the dynamics of agent $i$ in the network is described by $\dot{x}_i(t) = \tilde{u}_i(t) + g_i(t)$ with $\|g_i(t)-g_j(t)\|=0$ for all $i,j\in\mathcal{V}$. Moreover, $\tilde{u}_i$ has a form similar to $\omega_i$. Thus, from Lemma~\ref{lemma:app}, it follows that there exists a $T_\text{con}>0$ such that $x_i(t)=\dfrac{1}{N}\sum_{j=1}^Nx_j(t)$ for $t\geq T_\text{p}+T_\text{con}$, where $T_\text{con}$ satisfies $
	T_\text{con} \leq \frac{2}{q\alpha N^{2(1-\tau_1)}\tilde{c}^{\tau_1}(\tau_1-1)}+\frac{2}{q\beta \tilde{c}^{\tau_2}(1-\tau_2)}$, where $\tau_1\coloneqq \frac{1+\mu_1}{2}$, $\tau_2\coloneqq \frac{1+\mu_2}{2}$, and $\tilde{c}>0$ is an appropriate constant.
\end{proof}

Finally, the following result establishes that the agents track the optimal point in a fixed time.

\begin{Theorem}[\textbf{Fixed-time distributed optimization}]\label{cor:distributed_opt new}
Let each agent $i\in\mathcal{V}$ in the network be driven by the control input $u_i$ \eqref{eq:u_i new}. If the functions satisfy Assumption \ref{ass:2 new} (respectively, Assumption \ref{f assum 2}), then the agents track the minimizer of the team objective function within fixed time \textnormal{$T = T_\text{p}+T_\text{con}+T_{sc}$} (respectively, \textnormal{$T = T_\text{p}+T_\text{con}+T_{PL}$}). 
\end{Theorem}

\begin{proof}
The proof follows directly from the previous results presented in this section. From Lemmas~\ref{lemma:app} and \ref{thm:consensus new}, it follows that $g_i(t) = g_j(t)$ for all $t\geq T_\text{p}$, and $x_i(t) = x_j(t)$ for all $t\geq T_\text{p}+T_\text{con}$. Since $g_i$ is a function of $\theta_i$, and from Lemma~\ref{thm:param_estimation new}, we have that $\theta_i(t) = \sum_j\nabla f_j(x_j(t))$ for all $t\geq T_\text{p}$, with $\nabla f_i(x_i(t)) = \nabla f_j(x_j(t))$ for all $t\geq T_\text{p}+T_\text{con}$, we obtain that $g_i(t)=g^*(t)$ and $\tilde{u}_i(t)=0$ for all $i\in\mathcal{V}$, $t\geq T_\text{p}+T_\text{con}$.
Thus, if the objective functions satisfy Assumption \ref{ass:2 new} (respectively, Assumption \ref{f assum 2}), the conditions of the centralized fixed-time protocol in Lemma~\ref{thm:central new} are satisfied, and therefore, $x_i(t)=x^*$ for all $i\in\mathcal{V}$, for $t\geq T_\text{p}+T_\text{con}+T_{sc}$ (respectively, $t\geq T_\text{p}+T_\text{con}+T_{PL}$).
\end{proof}
Note that the total time of convergence $T = T_\text{p}+T_\text{con}+T_{sc}$ ((respectively, $t\geq T_\text{p}+T_\text{con}+T_{PL}$)) depends upon the design parameters and is inversely proportional to $p, q,\alpha, \beta, \gamma, \delta, \mu_1, \mu_2, \nu_1, \nu_2, l_1, l_2$. Hence, for a given user-defined time budget $T_b$, one can choose large values of these parameters so that $T\leq T_b$, and hence, convergence can be achieved within user-defined time $T_b$. Note that the time of convergence decreases as $\mu_1, l_1, \nu_1$ increase and $\mu_2, l_2, \nu_2$ decrease. The overall Fixed-time stable Distributed Optimization Algorithm (FxTS-DOA) with discrete-time iterative implementation is described in Algorithm~\ref{alg:opt}.

\begin{algorithm}
	\caption{Discretized FxTS-DOA}\label{alg:opt}
	\begin{algorithmic}[1]
		\Procedure{FxTS Dist Opt}{$(A,\mathcal{V})$, $\{f_i(\cdot)\}$}
		\State {\bf Inputs}: Parameters $p,q,l_1,l_2,\nu_1,\nu_2,\mu_1,\mu_2$;\quad  Step-size $\eta$
		\State Initialize local estimates $\{x_i\}$ for each $i\in\mathcal{V}$
		\State $\omega_i \gets 0_{d\times 1}$ for each $i\in\mathcal{V}$
		\State $\theta_i \gets 0_{d\times 1}$ for each $i\in\mathcal{V}$
		\For{\texttt{$k = 1$, $k\leq$max-epochs}}
            \State Each agent computes its own gradient $\nabla_if_i(x_i)$\vspace{1em}
            \State $\bar{u}_i\gets q\left((x_j\!-\!x_i)+\sign^{\mu_1}(x_j\!-\!x_i)+\sign^{\mu_2}(x_j\!-\!x_i)\right)$
            \State $\omega_i\gets \eta\left((\theta_j\!-\!\theta_i)+\sign^{\nu_1}(\theta_j\!-\!\theta_i)+\sign^{\nu_2}(\theta_j\!-\!\theta_i)\right)$
            \State \Comment{Information sharing with neighbors}\vspace{1em}
            \State $\theta_i\gets\omega_i + \eta\nabla_if_i(x_i)$
            \State $g_i\gets-\left((N\theta_i)+\sign^{l_1}(N\theta_i)+\sign^{l_2}(N\theta_i)\right)$
            \State $x_i\gets x_i+\eta (g_i+\bar{u}_i)$
            \State \Comment{Agents update their estimates locally}
        \EndFor
		\State \textbf{return} $x_1 = x_2 = \dots = x^*$
		\EndProcedure
	\end{algorithmic}
\end{algorithm}


\begin{Remark}
There may exist communication link failures or additions among generator buses, which results in a time-varying communication topology. We model the underlying graph $\mathcal G(t) = (A(t),\mathcal{V})$ through a time-varying signal $\rchi(t):\mathbb{R}_+\to\Psi$ as $\mathcal G(t)=\mathcal G_{\rchi(t)}\coloneqq(A_{\rchi(t)},\mathcal{V})$, where $\Psi=\{1,2,\dots,R\}$ is a finite set consisting of index numbers associated to specific adjacency matrices $A(t) = [a_{ij}(t)]\in \{A_1,\dots, A_R\}$. It can be easily shown that the proposed results extend to the case of time-varying topology under the condition that the graph $\mathcal G(t)$ is connected at all times. 
\end{Remark}

\section{Discretization of the FxTS-DOA}\label{sec: discretization}
Continuous-time dynamical systems, such as the one given by \eqref{sys flow} with $u_i$ given by \eqref{eq:u_i new}, offer effective insights into designing accelerated schemes for solving a distributed optimization problem. However, in practice, a discrete-time implementation is used for solving optimization problems. In general, the fixed-time convergent behavior of the FxTS-DOA does not need to be preserved upon discretization. A consistent discretization scheme preserves the convergence behavior of the continuous-time dynamical system in the discrete-time setting as well (see, e.g., \cite{polyakov2019consistent}). In particular, \cite{polyakov2019consistent} characterizes a discretization to be consistent with a fixed-time convergent dynamical system if the trajectories of the discretized system converge to an arbitrarily small neighborhood of the equilibrium point of the continuous-time system within a fixed number of steps, independent of the initial conditions. The analysis below shows that when the fixed-time convergent closed-loop dynamics \eqref{sys flow} under $u_i$ given by \eqref{eq:u_i new} is discretized using Euler discretization, it leads to consistent discretization.

In order to prove that an Euler discretization scheme of the proposed method in Section~\ref{sec:strong_convex} leads to a consistent discretization, it is sufficient to show that the closed-loop dynamics \eqref{sys flow} under $u_i$ given by \eqref{eq:u_i new} satisfies the conditions of \cite[Theorem 3]{garg2021MVIP}. Consider the proposed algorithm in Section~\ref{sec:strong_convex}. For $0\leq t\leq T_\text{p}+T_\text{con}$, the dynamics for all $\theta_i, x_i$ can be written in a compact form as:
\begin{align}
    \dot{\theta} & = F_1(\theta) + F_2(x),\quad
    \dot x = F_3(x) + F_4(\theta),
\end{align}
where {\small
\begin{align*}
    \hspace{-12pt}F_1(\theta)\!=\!\begin{bmatrix}\omega_1\\ \vdots \\ \omega_N \end{bmatrix}\!, F_2(x) \!=\! \begin{bmatrix}h_1\\ \vdots \\ h_N \end{bmatrix}\!,
    F_3(x)\!=\! \begin{bmatrix}\tilde u_1\\ \vdots \\ \tilde u_N \end{bmatrix}\!, F_4(\theta)\!=\! \begin{bmatrix}g_1\\ \vdots \\ g_N\end{bmatrix}\!.
\end{align*}}\normalsize
More compactly, define $z = [\theta^{\intercal},\; x^{\intercal}]^{\intercal}\in \mathbb R^{2Nd}$ and $\mathcal F(z) \coloneqq [(F_1(\theta) + F_2(x))^{\intercal}, \; F_3(x)^{\intercal}+F_4(\theta)^{\intercal}]^{\intercal}$ so that 
\begin{align}\label{eq: dot z}
    \dot z \in \mathcal F(z).
\end{align}
{We use the notion of differential inclusion in \eqref{eq: dot z} since the right-hand side of \eqref{eq: dot z} is not single-valued.} The interested reader is referred to \cite{clarke2008nonsmooth} for more details. First, we show that the set-valued map $\mathcal F$ in \eqref{eq: dot z} satisfies the conditions in \cite[Theorem 3]{garg2021MVIP}.

\begin{figure*}[!ht]
	\begin{center}
		\begin{tabular}{cc}
			\includegraphics[width=0.685\columnwidth]{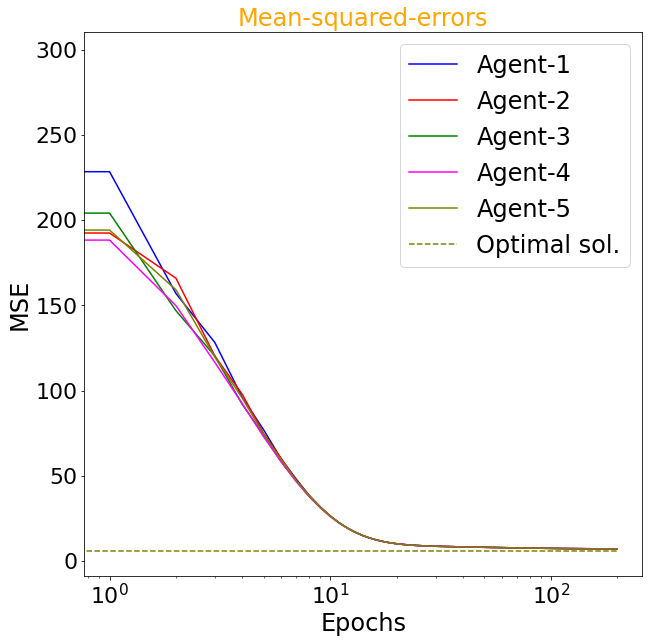}&
			\includegraphics[width=0.70\columnwidth]{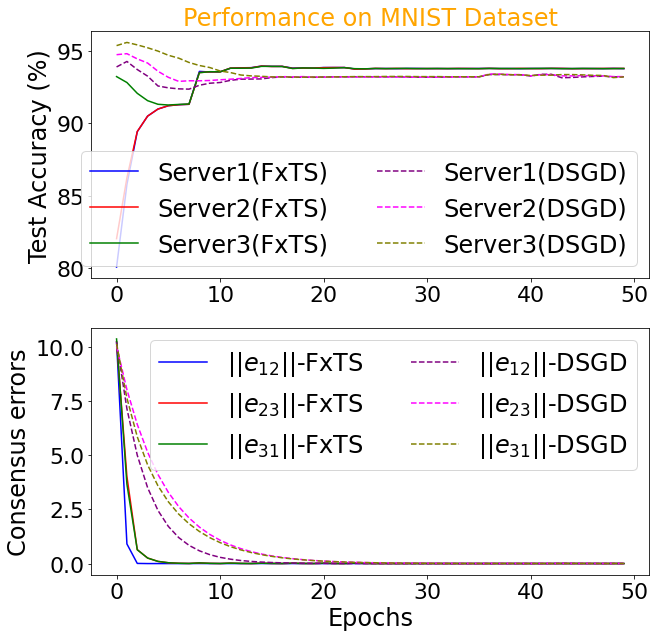}\cr
			(a)&(b)
		\end{tabular}
		\caption{Evaluation of the proposed DOA for (a) linear regression, (b) training of deep neural networks}
		\label{fig:Exp}
	\end{center}
	\vspace{-1em}
\end{figure*}

\begin{Lemma}\label{lemma: F assum satisfy}
If the functions $f_i$ satisfy either Assumption \ref{ass:2 new} (or Assumption \ref{f assum 2}) for all $i\in \mathcal{V}$, then $\mathcal F$ in \eqref{eq: dot z} is upper semi-continuous set-valued map, taking non-empty, convex and compact values. 
\end{Lemma}
\begin{proof}
Define $S = \{z\; |\; \mathcal F(z) = 0\}$ is the set of equilibrium points for the dynamics of variable $z$. Note that the equilibrium points of \eqref{eq: dot z} are the points $x_i = x_j$ and $\theta_i = \theta_j$ for all $i\neq j$, which is a $2d-$dimensional subspace in $\mathbb R^{2Nd}$, and thus, $S$ is a Lebesgue measure zero set in $\mathbb R^{2Nd}$. Note that the map $\mathcal F$ is continuous for all $z\in \mathbb R^{2Nd}\setminus \bigcup\limits_{i\neq j}S_{ij}$ where 
\begin{align}
    \hspace{-10pt}S_{ij} = \{z = [\theta^{\intercal}, \; x^{\intercal}]^{\intercal}\; |\; x_i = x_j, \theta_i = \theta_j\} \subset \mathbb R^{2Nd-2d},
\end{align}
and is also locally essentially bounded. From \cite[Remark 2]{danca2010chaotic}, we obtain that the map $\mathcal F$ is upper semi-continuous with non-empty, compact and convex values for all $0\leq t\leq T_\text{p}+T_\text{con}$. 

Now, it holds that $\omega_i(t) = 0$ and $\tilde u_i(t) = 0$ for $i = \{1, 2, \dots, N\}$, i.e., $F_1(\theta(t)) = F_3(x(t)) = 0$, for all $t\geq T_\text{p}+T_\text{con}$. Furthermore, for $t\geq T_\text{p}+T_\text{con}$, it holds that $g_i(\theta (t)) = g^*(x(t))$, and thus, $F_4(\theta(t)) = F_4(x(t))$. Hence, the augmented dynamics for $t\geq T_\text{p}+T_\text{con}$ reads:
\begin{align}
    \dot{\theta}(t) & = F_2(x(t)), \quad
    \dot x(t) = F_4(x(t)).
\end{align}
Note that $F_2$ and $F_4$ are continuous functions in their arguments, and thus, the map $\mathcal F(z)$ the required conditions for all $t\geq T_\text{p}+T_\text{con}$, which completes the proof. 
\end{proof}

Now, we are ready to present the main result of this section, which shows that when the closed-loop dynamics of \eqref{sys flow} under $u = \tilde u_i + g_i$, written compactly as \eqref{eq: dot z}, is discretized using Euler discretization, the trajectories of the resulting discrete-time system reach an arbitrarily small neighborhood of the optimal point $x^*$ within a fixed number of steps. To this end, define $z^*$ as $z^* \coloneqq \begin{bmatrix}I_N\otimes x^* \\ 0_{Nd} \end{bmatrix}$
where $\otimes$ denote the Kronecker product, $I_N\in \mathbb R^{N\times N}$ an identity matrix, and $0_{Nd}\in \mathbb R^{Nd}$ a vector consisting of zeros.

\begin{Theorem}\label{thm: discretized}
Assume that the functions $f_i$ satisfy Assumption~\ref{ass:2 new} (or Assumption~\ref{f assum 2}) for all $i\in \mathcal{V}$ and let $p = q$, $\alpha = \gamma$, $\beta  = \delta$, $l_1 = \mu_1 = \nu_1 = 1+ \frac{1}{\mu}$ and $l_2 = \mu_2 = \nu_2 = 1-\frac{1}{\mu}$ for some $\mu>1$. Consider the Euler discretization of \eqref{eq: dot z} given by
\begin{align}\label{eq:dot z disc}
    z_{k+1} \in z_k +\eta \mathcal F(z_k),
\end{align}
where $\eta>0$. Then, for each $\epsilon>0$, there exists $\eta^*>0$ such that for all $\eta\in (0, \eta^*]$, the trajectories of \eqref{eq:dot z disc} satisfy {\small
\begin{align}\label{eq:disc z star}
    \|z_k-z^*\|\leq \begin{cases}
   \frac{1}{\sqrt{c_1}}\left(\sqrt{\frac{a}{b}}\tan\left(\frac{\pi}{2}-\frac{\eta k\sqrt{ab}}{2\mu}\right)\right)^{\bar\mu} +\epsilon\; ; & k\leq \frac{\mu\pi}{\sqrt{ab}\eta}\\
    \epsilon \; ; & \textrm{otherwise},
    \end{cases},
\end{align}}\normalsize
where $a, b, c_1, \mu>0$
\end{Theorem}
\noindent The proof is provided in Appendix \ref{app: proof thm discretized}. 
Thus, it is shown that the trajectories of the closed-loop dynamics \eqref{sys flow} of each node $i$ under the input \eqref{eq:u_i new}, when discretized using Euler discretization, converge to an arbitrarily small neighborhood (dictated by $\epsilon$) of the optimal point $x^*$ within a fixed number of steps $\frac{\bar \mu}{2\sqrt{ab}\eta}$, independent of the initial conditions $x_i(0)$.

\section{Numerical Validation}\label{sec:simulation}
We now validate the proposed fixed-time convergent distributed optimization algorithm on two large-scale learning tasks. The algorithm was implemented using PyTorch 0.4.1 on a 16GB Core-i7 2.8GHz CPU and NVIDIA GeForce GTX-1060 GPU\footnote{source code will be made available in our public submission}.

The first task concerns distributed linear regression among a network of $N=5$ agents. The goal is to find the linear relationship between an input $z\in\mathbb{R}^3$ and an output $y\in\mathbb{R}^2$ as $y = Wz + b$. A dataset comprising $n=15,000$ points is randomly distributed across the five agents. Each agent can access its dataset and only exchange information with its immediate neighbor. Additionally, each agent has its own estimate of the parameter vectors $W\in\mathbb{R}^{2\times 3}$, and $b\in\mathbb{R}^{2}$, denoted by $x_i\coloneqq[W_i^\intercal, b_i^\intercal]^\intercal$. Here, $W_i$ and $b_i$ are vectorized representations of agent $i$'s estimate of the parameters $W$ and $b$. The agents find the regressor by minimizing the mean-squared-error on their respective loss functions, i.e.,
\begin{align*}
    f_i(x_i) = \frac{1}{n_i}\sum_{j=1}^{n_i}\|y_j-(W_iz_j+b_i)\|^2,
\end{align*}
where $n_i$ is the size of $i^\text{th}$-agent's dataset. Fig.~\ref{fig:Exp}a shows the performance of our algorithm with each epoch. Despite working with different data points and having different initial parameter estimates, the agents converge to the optimal solution in a very few epochs while reaching consensus on their estimates of the regression parameters.

We further validate the performance of the proposed DOA for distributed training of deep neural networks on the MNIST dataset~\cite{lecun1998gradient}. This is in contrast to All-Reduce algorithm~\cite{cho2019blueconnect}, where different servers carry the same parameter vector while exchanging gradient information with their neighbors in a ring topology. Instead, we assume a network of three servers connected in a line graph where each server has access to only one-third (20k) of the total (60k) training images. We consider a network with a single convolutional layer with ReLU activation (consisting of 32 filters of size $3\times 3$), followed by a dense layer (with ReLU activation) of output size 128. The final linear layer transforms 128-dimensional input to a 10-dimensional output (corresponding to 10 classes) with SoftMax activation. The network comprises a total of 2.8$\times 10^6$ learnable parameters. The individual servers have their own local estimates of the neural network parameters. Figure~\ref{fig:Exp}b shows that the servers initialized with different parameters and having different test accuracies quickly converge to around 94\% accuracy in less than ten epochs. Moreover, the norms of the consensus errors between servers $i$ and $j$, denoted by $e_{ij}\coloneqq x_i-x_j$, too, converge to zero, indicating that all the servers arrive at a similar estimate for all the neural network parameters. We also compare the performance of the proposed FxTS-DOA with the decentralized SGD (DSGD)~\cite{koloskova2020unified} algorithm. As can be seen in Figure~\ref{fig:Exp}b for the DSGD method, even though the servers have better initial test accuracies to start with, the non-agreement between initial parameter estimates and large consensus errors eventually drives the cumulative test accuracy to $\sim93.27\%$. Moreover, the servers achieve consensus on parameter estimates only after 20 epochs. On the other hand, the proposed FxTS-DOA trades off initial dip in test accuracies for super fast consensus on network parameters, eventually resulting in improved cumulative performance.

The above results are quite significant since both the optimization and consensus are achieved in less than 20 epochs. This is particularly important for distributed training of neural networks, where simultaneous consensus on parameter estimates of nearly 2.8$\times 10^6$ parameters and gradients of private objective functions are being achieved in a distributed manner. The exchange of local estimates of parameters and gradients between any two neighbors occurs only once per epoch, i.e., the iteration complexity is only linear in the number of parameters (see lines 8-9 in Algorithm~\ref{alg:opt}), resulting in significantly lower computational overhead. Most existing approaches to distributed learning, such as All-Reduce~\cite{cho2019blueconnect} or distributed-SGD~\cite{dai2021distributed} assume the same initial parameter estimates while relying on the exchange of global gradient vector for achieving distributed optimization. On the other hand, we assume a distributed framework where each agent starts with its own parameter estimate and exchanges information with neighbors to arrive at a consensus on parameter and gradient vectors. Thus, the FxTS-DOA does not have to wait for the consensus to occur on the global gradient vector before agents or servers can update their parameters.

\section{Conclusions}\label{sec:conclusions}
This paper presented a scheme to solve a distributed convex optimization problem for continuous-time multi-agent systems with fixed-time convergence guarantees under various conditions on the team-objective function. We showed that even when the communication topology of the network varies with time, consensus on the state values, as well as the gradient and the Hessian (if required) of the function values, can be achieved in a fixed time. It is shown that each aspect of the algorithm, the consensus on the crucial information and convergence on the optimal value, are achieved in a fixed time. Finally, it is shown that when the proposed continuous-time scheme is discretized using Euler's method, the fixed-time convergence properties are preserved. This is also verified through various numerical studies. Future work involves investigating distributed optimization methods with fixed-time convergence guarantees with private convex constraints.


\appendix

\section{Proof of Lemma~\ref{lemma:app}}\label{app: lemma app proof}
\begin{proof}
The time derivative of $\theta_i$ is given by:
\begin{align*}
	\dot \theta_i & = p\sum_{j\in \mathcal N_i}\Big(\sign(\theta_j-\theta_i) + \gamma \sign^{\nu_1}(\theta_j-\theta_i) \\
	&\qquad +\delta \sign(\theta_j-\theta_i)^{\nu_2}\Big) + h_i.
\end{align*}
Define $\theta_{ji}\coloneqq \theta_j-\theta_i$ and $\theta_c \coloneqq \frac{1}{N}\sum_{j=1}^N\theta_j$, $i,j\in \mathcal V$. The difference between an agent $i$'s state $\theta_{i}$ and the mean value $\theta_c$ of all agents' states is denote by $\tilde{\theta}_i\coloneqq \theta_i-\theta_c$. Similarly, $\tilde{\theta}_{ji}$ represents the difference $(\tilde{\theta}_j-\tilde{\theta}_i)$. The time-derivative of $\tilde{\theta}_{i}$ is given by:
\begin{align}\label{eq:theta_i_tilde_dot}
	\dot{\tilde{\theta}}_i &= \omega_i + h_i - \frac{1}{N}\sum_{j=1}^N\omega_j - \frac{1}{N}\sum_{j=1}^Nh_j \nonumber\\
	&= \frac{1}{N}\sum_{j=1}^N\left(\omega_i-\omega_j\right) + \frac{1}{N}\sum_{j=1}^N\left(h_i-h_j\right)
\end{align}
Define the error vector $\tilde{\theta} = \begin{bmatrix}\tilde{\theta}_1 & \tilde{\theta}_2 & \cdots \tilde{\theta}_N\end{bmatrix}^{\intercal}$. Consider the candidate Lyapunov function defined as $V(\tilde \theta) = \frac{1}{2}\|\tilde \theta\|^2 =  \frac{1}{2}\sum_{i=1}^N\tilde{\theta}_i^{\intercal}\tilde{\theta}_i$. 
Taking its time-derivative along the trajectories of \eqref{eq:theta_i_tilde_dot} yields:
\begin{align}\label{eq:V1_V2}
	\dot{V} (\tilde \theta)
	&= \underbrace{\frac{1}{N}\sum\limits_{i,j=1}^N\tilde{\theta}_i^{\intercal}\left(\omega_i-\omega_j\right)}_{\dot{V}_1} + \underbrace{\frac{1}{N}\sum\limits_{i,j=1}^N\tilde{\theta}_i^{\intercal}\left(h_i-h_j\right)}_{\dot{V}_2}.
\end{align}
From \eqref{eq:omega_dot new}, the first term $\dot{V}_1$ is rewritten as:\small{
\begin{align}\label{eq:V1}
	\dot{V}_1 &= \frac{1}{N}\sum_{i=1}^N\tilde{\theta}_i^{\intercal}\sum_{j=1}^N\Big(\omega_i-p\sum_{k\in \mathcal{N}_j}\Big(\sign(\theta_k-\theta_j) \nonumber\\
	&\qquad + \gamma \sign^{\nu_1}(\theta_k-\theta_j) +\delta \sign(\theta_k-\theta_j)^{\nu_2}\Big)\Big) \nonumber\\
	&\overset{\eqref{sign 0}}{=} \frac{1}{N}\sum_{i=1}^N\tilde{\theta}_i^{\intercal}\sum_{j=1}^N\omega_i  = \sum_{i=1}^N\tilde{\theta}_i^{\intercal}\omega_i  = p\sum_{i=1}^N\tilde{\theta}_i^{\intercal}\sum_{j\in\mathcal{N}_i}\left(\sign(\theta_{ji})\right. + \nonumber\\
	&\qquad \left.\gamma \sign^{\nu_1}(\theta_{ji}) +\delta \sign(\theta_{ji})^{\nu_2}\right)\nonumber\\
	&\overset{\eqref{f ij  eij}}{=} \frac{p}{2}\sum_{i=1}^N\sum_{j\in\mathcal{N}_i}\tilde{\theta}_{ij}^{\intercal}\left(\sign(\theta_{ji}) + \gamma \sign^{\nu_1}(\theta_{ji}) +\delta \sign(\theta_{ji})^{\nu_2}\right) \nonumber,
\end{align}}\normalsize
where the last equality follows with $w(x) = x$ in \eqref{f ij  eij}. Using this, and the fact that $\sign(\theta_{ij})^l = -\sign(\theta_{ji})^l$ for any odd $l\geq 0$, we obtain\small{
\begin{align}
\dot V_1 & = -\frac{p}{2}\sum_{i=1}^N\sum_{j\in\mathcal{N}_i}\tilde{\theta}_{ij}^{\intercal}\left(\sign(\theta_{ij}) + \gamma \sign^{\nu_1}(\theta_{ij}) +\delta \sign(\theta_{ij})^{\nu_2}\right) \nonumber\\
	&= -\frac{p}{2}\sum_{i=1}^N\sum_{j\in\mathcal{N}_i}\left(\|\tilde{\theta}_{ij}\| + \gamma\|\tilde{\theta}_{ij}\|^{\nu_1+1} + \delta\|\tilde{\theta}_{ij}\|^{\nu_2+1}\right),
\end{align}}\normalsize
where the last equality follows from $\tilde{\theta}_{ij}=(\theta_{i}-\theta_{c})-(\theta_{j}-\theta_{c})=\theta_{ij}$. The second term in \eqref{eq:V1_V2} can be bounded as:
\begin{align}\label{eq:V2}
	\dot{V}_2 &= \frac{1}{2N}\!\sum_{i,j=1}^N\!\!\tilde{\theta}_{ij}^{\intercal}\left(h_i-h_j\right) \leq \frac{1}{2N}\!\sum_{i,j=1}^N\!\|\tilde{\theta}_{ij}\| \|h_i-h_j\| \nonumber\\
	&\leq \frac{\rho}{2N}\sum_{i,j=1}^N\|\tilde{\theta}_{ij}\| \leq \frac{\rho}{2N}\left(N\max_{i}\sum_{j=1, j\neq i}^N\|\tilde{\theta}_{ij}\|\right)\nonumber\\
	&\leq \frac{\rho}{2}\frac{(N-1)}{2} \sum_{i=1}^N\sum_{j\in\mathcal{N}_i}\|\tilde{\theta}_{ij}\|,
\end{align}
where the last inequality follows from connectivity of $\mathcal G$. 
Thus, from \eqref{eq:V1} and \eqref{eq:V2}, it follows that\small{
\begin{align*}
	\dot{V}(\tilde\theta) &\leq -\frac{1}{2}\left(p-\rho\frac{(N-1)}{2}\right)\sum_{i=1}^N\sum_{j \in \mathcal N_i}\|\tilde{\theta}_{ij}\|  \nonumber\\
	&\qquad -\frac{1}{2}p\gamma \sum_{i=1}^N\sum_{j \in \mathcal N_i}\|\tilde{\theta}_{ij}\|^{\nu_1+1} - \frac{1}{2}p\delta \sum_{i=1}^N\sum_{j \in \mathcal N_i}\|\tilde{\theta}_{ij}\|^{\nu_2+1} \nonumber \\
	&\leq - \frac{1}{2}p\gamma \sum_{i=1}^N\sum_{j \in \mathcal N_i}\|\tilde{\theta}_{ij}\|^{\nu_1+1} - \frac{1}{2}p\delta \sum_{i=1}^N\sum_{j \in \mathcal N_i}\|\tilde{\theta}_{ij}\|^{\nu_2+1} \nonumber \\
	&\leq - \frac{p\gamma}{2}\sum_{i=1}^N\sum_{j=1}^N\left(a_{ij}\|\tilde{\theta}_{ij}\|^2\right)^{\kappa_1} - \frac{p\delta}{2}\sum_{i=1}^N\sum_{j=1}^N\left(a_{ij}\|\tilde{\theta}_{ij}\|^2\right)^{\kappa_2},
\end{align*}}\normalsize
where $\kappa_1=\frac{1+\nu_1}{2}$, $\kappa_2=\frac{1+\nu_2}{2}$. Define $\eta_{ij}=a_{ij}\|\tilde{\theta}_{ij}\|^2$. With this, and using the fact that $\nu_1>1$ and $\nu_2<1$, we obtain:
\begin{align*}
	\dot{V}(\tilde\theta) &\leq -\frac{p\gamma}{2}\sum_{i=1}^N\sum_{j=1}^N\eta_{ij}^{\kappa_1} - \frac{p\delta}{2}\sum_{i=1}^N\sum_{j=1}^N\eta_{ij}^{\kappa_2} \nonumber\\
	&\overset{\eqref{p ineq}}{\leq} -\frac{p\gamma}{2}N^{2(1-\kappa_1)}\left(\sum_{i,j=1}^N\eta_{ij}\right)^{\kappa_1}\!\!-\frac{p\delta}{2}\left(\sum_{i,j=1}^N\eta_{ij}\right)^{\kappa_2}\!.
\end{align*}
We have 
\begin{align*}
	\sum_{i,j = 1}^N\eta_{ij} &= \sum_{i,j = 1}^Na_{ij}\|\tilde{\theta}_{ij}\|^2 \nonumber\\
	&= 2\tilde{\theta}^{\intercal}L_A\otimes I_{N}\tilde{\theta} \geq 2\lambda_2(L_A\otimes I_{N})\tilde{\theta}^{\intercal}\tilde{\theta} = c V, 
\end{align*}
where $c = 4\lambda_2(L_A)$. With this, we obtain that $$\dot V(\tilde\theta)  \leq-\frac{p\gamma}{2}N^{2(1-\kappa_1)}c^{\kappa_1}V(\tilde\theta)^{\kappa_1}-\frac{p\delta}{2}c^{\kappa_2}V(\tilde\theta)^{\kappa_2}.$$ With $\nu_1>1$, we have $\kappa_1>1$, and with $\nu_2<1$, we have $\kappa_2<1$. Hence, using Lemma~\ref{FxTS TH}, we obtain that $V(\tilde\theta(t)) = 0$, i.e., $\theta_i(t) = \theta_c(t)$, for all $t\geq T_\text{p}$, where $T_\text{p} = \frac{2}{p\gamma N^{2(1-\kappa_1)}c^{\kappa_1}(\kappa_1-1)}+\frac{2}{p\delta c^{\kappa_2}(1-\kappa_2)}$. Using the fact that $\sum\limits_{i = 1}^N\omega_i(t) = 0$ for all $t\geq 0$, we obtain that 
\begin{align*}
    \hspace{-10pt} \sum_{i =1}^N\dot \theta_i(t) =  & \sum_{i =1}^N\omega_i(t) +  \sum_{i =1}^Nh_i(t) =  \sum_{i =1}^N h_i(t)=  \sum_{i =1}^N\frac{d}{dt}\zeta_i(t),\\
    \implies &  \sum_{i =1}^N\theta_i(t) =  \sum_{i =1}^N\zeta_i(t) + c.
\end{align*}
With $\theta_i(0)\!=\!\zeta_i(0)$, we have $c\!=\!0$, which completes the proof. 
\end{proof}
\vspace{-1.5em}
\section{Proof of Theorem~\ref{thm: discretized}}\label{app: proof thm discretized}
\begin{proof}
First, consider the closed-loop dynamics \eqref{eq: dot z} for $t\leq T_\text{p}+T_\text{con}$. From Lemma~\ref{thm:param_estimation new}, it holds that the function $V_1(\theta) = \frac{1}{2}\|\tilde \theta\|^2$, where $\tilde \theta$ is as defined in Lemma~\ref{lemma:app} satisfies $\dot V_1(\theta(t))\leq -a_1V_1(\theta)^{\kappa_1}-a_2V_1(\theta)^{\kappa_2}$, 
where $a_1 = \frac{p\gamma}{2}N^{2(1-\kappa_1)}c^{\kappa_1}, a_2 = \frac{p\delta}{2}c^{\kappa_2}$, $\kappa_1 = \frac{1+\nu_1}{2}>1$, $\kappa_2 = \frac{1+\nu_2}{2}<1$ and $c = 4\lambda_2(L_A)$. Similarly, since $p = q$, $\alpha = \gamma$, $\beta  = \delta$, $\mu_1 = \nu_1$ and $\mu_2 = \nu_2$, the function $V_2(x) = \frac{1}{2}\|\tilde x\|^2$ where $\tilde x_i(t) = x_i(t) - \frac{1}{N}\sum_{j = 1}^Nx_j(t)$, satisfies $\dot V_2(x(t))\leq -a_1V_2(x)^{\kappa_1}-a_2V_2(x)^{\kappa_2}$. Now, define $z = \begin{bmatrix}x\\ \theta\end{bmatrix}$ and $V(z(t)) = V_1(\theta(t)) + V_2(x(t))$, so that it holds that{
\begin{align*}
    \dot V(z(t))& \leq -a_1V_1(\theta)^{\kappa_1}-a_2V_1(\theta)^{\kappa_2} -a_1V_2(x)^{\kappa_1}-a_2V_2(x)^{\kappa_2}\\& 
    = -a_1(V_1(\theta)^{\kappa_1}+V_2(x)^{\kappa_1}) -a_2(V_1(\theta)^{\kappa_2}+V_2(x)^{\kappa_2}),  
\end{align*}}\normalsize
for all $t\leq T_\text{p}+T_\text{con}$. Now, using Lemma~\ref{lemma:ineq}, it holds that $V_1(\theta)^{\kappa_1} + V_2(x)^{\kappa_1} \geq 2^{1-\kappa_1}(V_1(\theta)+V_2(x))^{\kappa_1} = 2^{1-\kappa_1}V(z)^{\kappa_1}$ and $V_1(\theta)^{\kappa_2} + V_2(x)^{\kappa_2} \geq (V_1(\theta)+V_2(x))^{\kappa_2} = V(z)^{\kappa_2}$. Thus, it holds that $ \dot V(z(t)) \leq -a_12^{1-\kappa_1}V(z(t))^{\kappa_1}-a_2V(z(t))^{\kappa_2}$ for all $t\leq T_\text{p}+T_\text{con}$. Hence, $z = \bar z$ is an FxTS equilibrium point of \eqref{eq: dot z} where $
\bar z = \begin{bmatrix}I_N\otimes\bar \theta\\I_N\otimes \bar  x\\\end{bmatrix}$ with $\bar \theta = \frac{1}{N}\sum_{j = 1}^N\theta_j$ and $\bar x = \frac{1}{N}\sum_{j = 1}^Nx_j$. 
Note also that $V_1$ and $V_2$ are quadratic in the error $\tilde \theta$ and $\tilde x$, respectively. Hence, it holds that $V(z(t))  = \frac{1}{2}\|\tilde \theta\|^2 + \frac{1}{2}\|\tilde x\|^2 \leq \frac{1}{2}\|\tilde z\|^2$ where $\tilde z = z-z = \begin{bmatrix}\tilde \theta \\ \tilde x\end{bmatrix}$. Furthermore, it holds that $V(z(t))  = \frac{1}{2}\|\tilde \theta(t)\|^2 + \frac{1}{2}\|\tilde x(t)\|^2\geq \frac{1}{4}\|\tilde z(t)\|^2$.
Now, consider the time interval $t\geq T_\text{p}+T_\text{con}$. The dynamics for $z$ reads $
\dot z(t) = F(z(t)) = \begin{bmatrix}F_2(x(t)) \\ F_4(x(t))\end{bmatrix}$, for $t\geq T_\text{p}+T_\text{con}$. Consider the Lyapunov candidate $V_3(z(t)) = \frac{1}{2}\sum_{i= 1}^N\|\nabla f_i(x(t))\|^2 + \frac{1}{2}\|\tilde \theta(t)\|^2$. Note that from Lemma~\ref{thm:param_estimation new}, it follows that $\tilde\theta(t) = 0$ for $t\geq T_\text{p}$. If $f_i$ satisfies Assumption~\ref{ass:2 new}, then from Lemma~\ref{thm:central new}, it follows that the time derivative of $V$ along the trajectories of $z$ for $t\geq \max\{T_1, T_2\}$ reads
\begin{align*}
    \dot V_3(z(t))& \leq -k2^\frac{1+l_1}{2}V_3(z(t))^{\frac{1+l_1}{2}}-k2^\frac{1+l_2}{2}V_3(z(t))^{\frac{1+l_2}{2}}\\
    & = -k2^{\kappa_1}V_3(z(t))^{\kappa_1}-k2^{\kappa_2}V_3(z(t))^{\kappa_2},
\end{align*}
since $l_1 = \mu_1$ and $l_2 = \mu_2$. Note that from \cite[Theorem 2]{karimi2016linear}, it follows that under strong convexity implies quadratic growth, and thus, we obtain that the function $V$ satisfies the quadratic growth requirement in \cite[Theorem 3]{garg2021MVIP}.
If, on the other hand, $f_i$ satisfies Assumption \ref{f assum 2}, then from Lemma~\ref{corollary central PL}, it follows that the time derivative of $V$ along the trajectories of $z$ for $t\geq \max\{T_1, T_2\}$ reads
\begin{align*}
    \dot V_3(z(t))& \leq -k_1^\frac{1+l_1}{2}V_3(z(t))^{\frac{1+l_1}{2}}-k_2^\frac{1+l_2}{2}V_3(z(t))^{\frac{1+l_2}{2}}\\
    & \leq -k_1^{\kappa_2}V_3(z(t))^{\kappa_1}-k_2^{\kappa_2}V_3(z(t))^{\kappa_2}
\end{align*}
where $k_1 = (2\mu)^{\frac{1+l_1}{2}}2^\frac{1+l_1}{2}$ and $k_2 = (2\mu)^{\frac{1+l_2}{2}}2^\frac{1+l_2}{2}$.  Choose $a = \min\{a_12^{1-\kappa_1}, k2^{\kappa_1}, k_1^{\kappa_1}\}$ and $b = \min\{a_2, k2^{\kappa_2}, k_2^{\kappa_2}\}$, so that it holds that $\dot V(z(t)) \leq -aV(z(t))^{\kappa_1}-bV(z(t))^{\kappa_2}$ for all $t\geq 0$. In this case as well, since the system trajectories evolve in a compact set $\{z\; |\; V(z)\leq V(z(0))\}$, from \cite[Theorem 2]{karimi2016linear}, it follows that the function $V$ satisfies the quadratic growth requirement in \cite[Theorem 3]{garg2021MVIP}. Thus, all the conditions of \cite[Theorem 3]{garg2021MVIP} are satisfied with $\beta = \frac{1}{2}$, and hence, \eqref{eq:disc z star} holds. 
\end{proof}

\bibliographystyle{ACM-Reference-Format}
\bibliography{myreferences}


\begin{thebibliography}{30}


\ifx \showCODEN    \undefined \def \showCODEN     #1{\unskip}     \fi
\ifx \showDOI      \undefined \def \showDOI       #1{#1}\fi
\ifx \showISBNx    \undefined \def \showISBNx     #1{\unskip}     \fi
\ifx \showISBNxiii \undefined \def \showISBNxiii  #1{\unskip}     \fi
\ifx \showISSN     \undefined \def \showISSN      #1{\unskip}     \fi
\ifx \showLCCN     \undefined \def \showLCCN      #1{\unskip}     \fi
\ifx \shownote     \undefined \def \shownote      #1{#1}          \fi
\ifx \showarticletitle \undefined \def \showarticletitle #1{#1}   \fi
\ifx \showURL      \undefined \def \showURL       {\relax}        \fi
\providecommand\bibfield[2]{#2}
\providecommand\bibinfo[2]{#2}
\providecommand\natexlab[1]{#1}
\providecommand\showeprint[2][]{arXiv:#2}

\bibitem[Baranwal et~al\mbox{.}(2020)]%
        {baranwal2020robust}
\bibfield{author}{\bibinfo{person}{Mayank Baranwal}, \bibinfo{person}{Kunal
  Garg}, \bibinfo{person}{Dimitra Panagou}, {and} \bibinfo{person}{Alfred~O
  Hero}.} \bibinfo{year}{2020}\natexlab{}.
\newblock \showarticletitle{Robust Distributed Fixed-Time Economic Dispatch
  Under Time-Varying Topology}.
\newblock \bibinfo{journal}{\emph{IEEE Control Systems Letters}}
  \bibinfo{volume}{5}, \bibinfo{number}{4} (\bibinfo{year}{2020}),
  \bibinfo{pages}{1183--1188}.
\newblock


\bibitem[Benosman et~al\mbox{.}(2020)]%
        {benosman2020optimizing}
\bibfield{author}{\bibinfo{person}{Mouhacine Benosman},
  \bibinfo{person}{Orlando Romero}, {and} \bibinfo{person}{Anoop Cherian}.}
  \bibinfo{year}{2020}\natexlab{}.
\newblock \bibinfo{title}{Optimizing deep neural networks via discretization of
  finite-time convergent flows}.
\newblock
\newblock
\newblock
\shownote{arXiv e-Print}.


\bibitem[Bhat and Bernstein(2000)]%
        {bhat2000finite}
\bibfield{author}{\bibinfo{person}{Sanjay~P Bhat} {and}
  \bibinfo{person}{Dennis~S Bernstein}.} \bibinfo{year}{2000}\natexlab{}.
\newblock \showarticletitle{Finite-time stability of continuous autonomous
  systems}.
\newblock \bibinfo{journal}{\emph{SICON}} \bibinfo{volume}{38},
  \bibinfo{number}{3} (\bibinfo{year}{2000}), \bibinfo{pages}{751--766}.
\newblock


\bibitem[Cho et~al\mbox{.}(2019)]%
        {cho2019blueconnect}
\bibfield{author}{\bibinfo{person}{Minsik Cho}, \bibinfo{person}{Ulrich
  Finkler}, \bibinfo{person}{David Kung}, {and} \bibinfo{person}{Hillery
  Hunter}.} \bibinfo{year}{2019}\natexlab{}.
\newblock \showarticletitle{Blueconnect: Decomposing all-reduce for deep
  learning on heterogeneous network hierarchy}.
\newblock \bibinfo{journal}{\emph{Proceedings of Machine Learning and Systems}}
   \bibinfo{volume}{1} (\bibinfo{year}{2019}), \bibinfo{pages}{241--251}.
\newblock


\bibitem[Clarke et~al\mbox{.}(2008)]%
        {clarke2008nonsmooth}
\bibfield{author}{\bibinfo{person}{Francis~H Clarke}, \bibinfo{person}{Yuri~S
  Ledyaev}, \bibinfo{person}{Ronald~J Stern}, {and} \bibinfo{person}{Peter~R
  Wolenski}.} \bibinfo{year}{2008}\natexlab{}.
\newblock \bibinfo{booktitle}{\emph{Nonsmooth analysis and control theory}}.
  Vol.~\bibinfo{volume}{178}.
\newblock \bibinfo{publisher}{Springer Science \& Business Media}.
\newblock


\bibitem[Dai et~al\mbox{.}(2021)]%
        {dai2021distributed}
\bibfield{author}{\bibinfo{person}{LingFei Dai}, \bibinfo{person}{Boyu Diao},
  \bibinfo{person}{Chao Li}, {and} \bibinfo{person}{Yongjun Xu}.}
  \bibinfo{year}{2021}\natexlab{}.
\newblock \showarticletitle{A Distributed SGD Algorithm with Global Sketching
  for Deep Learning Training Acceleration}.
\newblock \bibinfo{journal}{\emph{arXiv preprint arXiv:2108.06004}}
  (\bibinfo{year}{2021}).
\newblock


\bibitem[Danca(2010)]%
        {danca2010chaotic}
\bibfield{author}{\bibinfo{person}{Marius-F Danca}.}
  \bibinfo{year}{2010}\natexlab{}.
\newblock \showarticletitle{Chaotic behavior of a class of discontinuous
  dynamical systems of fractional-order}.
\newblock \bibinfo{journal}{\emph{Nonlinear Dynamics}} \bibinfo{volume}{60},
  \bibinfo{number}{4} (\bibinfo{year}{2010}), \bibinfo{pages}{525--534}.
\newblock


\bibitem[Feng and Hu(2017)]%
        {feng2017finite}
\bibfield{author}{\bibinfo{person}{Zhi Feng} {and} \bibinfo{person}{Guoqiang
  Hu}.} \bibinfo{year}{2017}\natexlab{}.
\newblock \showarticletitle{Finite-time distributed optimization with quadratic
  objective functions under uncertain information}. In
  \bibinfo{booktitle}{\emph{IEEE 56th Annual Conference on Decision and
  Control}}. IEEE, \bibinfo{pages}{208--213}.
\newblock


\bibitem[Garg et~al\mbox{.}(2021)]%
        {garg2021MVIP}
\bibfield{author}{\bibinfo{person}{Kunal Garg}, \bibinfo{person}{Mayank
  Baranwal}, \bibinfo{person}{Rohit Gupta}, {and} \bibinfo{person}{Mouhacine
  Benosman}.} \bibinfo{year}{2021}\natexlab{}.
\newblock \bibinfo{title}{Fixed-Time Stable Proximal Dynamical System for
  Solving MVIPs}.
\newblock
\newblock
\showeprint{1908.03517}~[math.OC]
\newblock
\shownote{arXiv e-Print}.


\bibitem[Garg et~al\mbox{.}(2020)]%
        {garg2020fixed}
\bibfield{author}{\bibinfo{person}{Kunal Garg}, \bibinfo{person}{Mayank
  Baranwal}, {and} \bibinfo{person}{Dimitra Panagou}.}
  \bibinfo{year}{2020}\natexlab{}.
\newblock \showarticletitle{A fixed-time convergent distributed algorithm for
  strongly convex functions in a time-varying network}. In
  \bibinfo{booktitle}{\emph{2020 59th IEEE Conference on Decision and Control
  (CDC)}}. IEEE, \bibinfo{pages}{4405--4410}.
\newblock


\bibitem[Garg and Panagou(2020)]%
        {garg2018fixed}
\bibfield{author}{\bibinfo{person}{Kunal Garg} {and} \bibinfo{person}{Dimitra
  Panagou}.} \bibinfo{year}{2020}\natexlab{}.
\newblock \showarticletitle{Fixed-time stable gradient flows: Applications to
  continuous-time optimization}.
\newblock \bibinfo{journal}{\emph{IEEE Trans. Automat. Control}}
  \bibinfo{volume}{66}, \bibinfo{number}{5} (\bibinfo{year}{2020}),
  \bibinfo{pages}{2002--2015}.
\newblock


\bibitem[Hassan-Moghaddam and Jovanovi{\'c}(2021)]%
        {hassan2019proximal}
\bibfield{author}{\bibinfo{person}{Sepideh Hassan-Moghaddam} {and}
  \bibinfo{person}{Mihailo~R Jovanovi{\'c}}.} \bibinfo{year}{2021}\natexlab{}.
\newblock \showarticletitle{Proximal gradient flow and Douglas--Rachford
  splitting dynamics: global exponential stability via integral quadratic
  constraints}.
\newblock \bibinfo{journal}{\emph{Automatica}}  \bibinfo{volume}{123}
  (\bibinfo{year}{2021}), \bibinfo{pages}{109311}.
\newblock


\bibitem[Hu and Shao(2016)]%
        {hu2016smooth}
\bibfield{author}{\bibinfo{person}{Qinglei Hu} {and} \bibinfo{person}{Xiaodong
  Shao}.} \bibinfo{year}{2016}\natexlab{}.
\newblock \showarticletitle{Smooth finite-time fault-tolerant attitude tracking
  control for rigid spacecraft}.
\newblock \bibinfo{journal}{\emph{Aerospace Science and Technology}}
  \bibinfo{volume}{55} (\bibinfo{year}{2016}), \bibinfo{pages}{144--157}.
\newblock


\bibitem[Hu and Yang(2018)]%
        {hu2018distributed}
\bibfield{author}{\bibinfo{person}{Zilun Hu} {and} \bibinfo{person}{Jianying
  Yang}.} \bibinfo{year}{2018}\natexlab{}.
\newblock \showarticletitle{Distributed finite-time optimization for second
  order continuous-time multiple agents systems with time-varying cost
  function}.
\newblock \bibinfo{journal}{\emph{Neurocomputing}}  \bibinfo{volume}{287}
  (\bibinfo{year}{2018}), \bibinfo{pages}{173--184}.
\newblock


\bibitem[Karimi et~al\mbox{.}(2016)]%
        {karimi2016linear}
\bibfield{author}{\bibinfo{person}{Hamed Karimi}, \bibinfo{person}{Julie
  Nutini}, {and} \bibinfo{person}{Mark Schmidt}.}
  \bibinfo{year}{2016}\natexlab{}.
\newblock \showarticletitle{Linear Convergence of Gradient and
  Proximal-gradient Methods under the {P}olyak- {\L}ojasiewicz condition}. In
  \bibinfo{booktitle}{\emph{Joint European Conference on Machine Learning and
  Knowledge Discovery in Databases}}. Springer, \bibinfo{pages}{795--811}.
\newblock


\bibitem[Koloskova et~al\mbox{.}(2020)]%
        {koloskova2020unified}
\bibfield{author}{\bibinfo{person}{Anastasia Koloskova},
  \bibinfo{person}{Nicolas Loizou}, \bibinfo{person}{Sadra Boreiri},
  \bibinfo{person}{Martin Jaggi}, {and} \bibinfo{person}{Sebastian Stich}.}
  \bibinfo{year}{2020}\natexlab{}.
\newblock \showarticletitle{A unified theory of decentralized sgd with changing
  topology and local updates}. In \bibinfo{booktitle}{\emph{International
  Conference on Machine Learning}}. PMLR, \bibinfo{pages}{5381--5393}.
\newblock


\bibitem[LeCun et~al\mbox{.}(1998)]%
        {lecun1998gradient}
\bibfield{author}{\bibinfo{person}{Yann LeCun}, \bibinfo{person}{L{\'e}on
  Bottou}, \bibinfo{person}{Yoshua Bengio}, {and} \bibinfo{person}{Patrick
  Haffner}.} \bibinfo{year}{1998}\natexlab{}.
\newblock \showarticletitle{Gradient-based learning applied to document
  recognition}.
\newblock \bibinfo{journal}{\emph{Proc. IEEE}} \bibinfo{volume}{86},
  \bibinfo{number}{11} (\bibinfo{year}{1998}), \bibinfo{pages}{2278--2324}.
\newblock


\bibitem[Lin et~al\mbox{.}(2017)]%
        {lin2017distributed}
\bibfield{author}{\bibinfo{person}{Peng Lin}, \bibinfo{person}{Wei Ren}, {and}
  \bibinfo{person}{Jay~A Farrell}.} \bibinfo{year}{2017}\natexlab{}.
\newblock \showarticletitle{Distributed continuous-time optimization:
  nonuniform gradient gains, finite-time convergence, and convex constraint
  set}.
\newblock \bibinfo{journal}{\emph{IEEE Trans. Automat. Control}}
  \bibinfo{volume}{62}, \bibinfo{number}{5} (\bibinfo{year}{2017}),
  \bibinfo{pages}{2239--2253}.
\newblock


\bibitem[Lu and Tang(2012)]%
        {lu2012zero}
\bibfield{author}{\bibinfo{person}{Jie Lu} {and} \bibinfo{person}{Choon~Yik
  Tang}.} \bibinfo{year}{2012}\natexlab{}.
\newblock \showarticletitle{Zero-gradient-sum algorithms for distributed convex
  optimization: The continuous-time case}.
\newblock \bibinfo{journal}{\emph{IEEE Trans. Automat. Control}}
  \bibinfo{volume}{57}, \bibinfo{number}{9} (\bibinfo{year}{2012}),
  \bibinfo{pages}{2348--2354}.
\newblock


\bibitem[Mesbahi and Egerstedt(2010)]%
        {mesbahi2010graph}
\bibfield{author}{\bibinfo{person}{Mehran Mesbahi} {and}
  \bibinfo{person}{Magnus Egerstedt}.} \bibinfo{year}{2010}\natexlab{}.
\newblock \bibinfo{booktitle}{\emph{Graph Theoretic Methods in Multiagent
  Networks}}. Vol.~\bibinfo{volume}{33}.
\newblock \bibinfo{publisher}{Princeton University Press}.
\newblock


\bibitem[Nathan and Klabjan(2017)]%
        {nathan2017optimization}
\bibfield{author}{\bibinfo{person}{Alexandros Nathan} {and}
  \bibinfo{person}{Diego Klabjan}.} \bibinfo{year}{2017}\natexlab{}.
\newblock \showarticletitle{Optimization for large-scale machine learning with
  distributed features and observations}. In
  \bibinfo{booktitle}{\emph{International Conference on Machine Learning and
  Data Mining in Pattern Recognition}}. Springer, \bibinfo{pages}{132--146}.
\newblock


\bibitem[Nedic and Olshevsky(2015)]%
        {nedic2015distributed}
\bibfield{author}{\bibinfo{person}{Angelia Nedic} {and} \bibinfo{person}{Alex
  Olshevsky}.} \bibinfo{year}{2015}\natexlab{}.
\newblock \showarticletitle{Distributed optimization over time-varying directed
  graphs}.
\newblock \bibinfo{journal}{\emph{IEEE Trans. Automat. Control}}
  \bibinfo{volume}{60}, \bibinfo{number}{3} (\bibinfo{year}{2015}),
  \bibinfo{pages}{601--615}.
\newblock


\bibitem[Pan et~al\mbox{.}(2018)]%
        {pan2018distributed}
\bibfield{author}{\bibinfo{person}{Xiaowei Pan}, \bibinfo{person}{Zhongxin
  Liu}, {and} \bibinfo{person}{Zengqiang Chen}.}
  \bibinfo{year}{2018}\natexlab{}.
\newblock \showarticletitle{Distributed Optimization with Finite-Time
  Convergence via Discontinuous Dynamics}. In \bibinfo{booktitle}{\emph{2018
  37th Chinese Control Conference (CCC)}}. IEEE, \bibinfo{pages}{6665--6669}.
\newblock


\bibitem[Polyakov(2012)]%
        {polyakov2012nonlinear}
\bibfield{author}{\bibinfo{person}{Andrey Polyakov}.}
  \bibinfo{year}{2012}\natexlab{}.
\newblock \showarticletitle{Nonlinear feedback design for fixed-time
  stabilization of linear control systems}.
\newblock \bibinfo{journal}{\emph{IEEE Trans. Automat. Control}}
  \bibinfo{volume}{57}, \bibinfo{number}{8} (\bibinfo{year}{2012}),
  \bibinfo{pages}{2106}.
\newblock


\bibitem[Polyakov et~al\mbox{.}(2019)]%
        {polyakov2019consistent}
\bibfield{author}{\bibinfo{person}{Andrey Polyakov}, \bibinfo{person}{Denis
  Efimov}, {and} \bibinfo{person}{Bernard Brogliato}.}
  \bibinfo{year}{2019}\natexlab{}.
\newblock \showarticletitle{{C}onsistent {d}iscretization of {f}inite-time and
  {f}ixed-time {s}table {s}ystems}.
\newblock \bibinfo{journal}{\emph{SIAM Journal on {C}ontrol and
  {O}ptimization}} \bibinfo{volume}{57}, \bibinfo{number}{1}
  (\bibinfo{year}{2019}), \bibinfo{pages}{78--103}.
\newblock


\bibitem[Rabbat and Nowak(2004)]%
        {rabbat2004distributed}
\bibfield{author}{\bibinfo{person}{Michael Rabbat} {and}
  \bibinfo{person}{Robert Nowak}.} \bibinfo{year}{2004}\natexlab{}.
\newblock \showarticletitle{Distributed optimization in sensor networks}. In
  \bibinfo{booktitle}{\emph{Proceedings of the 3rd international symposium on
  Information processing in sensor networks}}. ACM, \bibinfo{pages}{20--27}.
\newblock


\bibitem[Rahili and Ren(2016)]%
        {rahili2016distributed}
\bibfield{author}{\bibinfo{person}{Salar Rahili} {and} \bibinfo{person}{Wei
  Ren}.} \bibinfo{year}{2016}\natexlab{}.
\newblock \showarticletitle{Distributed continuous-time convex optimization
  with time-varying cost functions}.
\newblock \bibinfo{journal}{\emph{IEEE Trans. Automat. Control}}
  \bibinfo{volume}{62}, \bibinfo{number}{4} (\bibinfo{year}{2016}),
  \bibinfo{pages}{1590--1605}.
\newblock


\bibitem[Su et~al\mbox{.}(2014)]%
        {su2014differential}
\bibfield{author}{\bibinfo{person}{Weijie Su}, \bibinfo{person}{Stephen Boyd},
  {and} \bibinfo{person}{Emmanuel Candes}.} \bibinfo{year}{2014}\natexlab{}.
\newblock \showarticletitle{A differential equation for modeling {N}esterov’s
  accelerated gradient method: Theory and insights}. In
  \bibinfo{booktitle}{\emph{Advances in Neural Information Processing
  Systems}}. \bibinfo{pages}{2510--2518}.
\newblock


\bibitem[Wang et~al\mbox{.}(2020)]%
        {wang2020distributed}
\bibfield{author}{\bibinfo{person}{Xiangyu Wang}, \bibinfo{person}{Guodong
  Wang}, {and} \bibinfo{person}{Shihua Li}.} \bibinfo{year}{2020}\natexlab{}.
\newblock \showarticletitle{A distributed fixed-time optimization algorithm for
  multi-agent systems}.
\newblock \bibinfo{journal}{\emph{Automatica}}  \bibinfo{volume}{122}
  (\bibinfo{year}{2020}), \bibinfo{pages}{109289}.
\newblock


\bibitem[Zuo and Tie(2016)]%
        {zuo2016distributed}
\bibfield{author}{\bibinfo{person}{Zongyu Zuo} {and} \bibinfo{person}{Lin
  Tie}.} \bibinfo{year}{2016}\natexlab{}.
\newblock \showarticletitle{Distributed robust finite-time nonlinear consensus
  protocols for multi-agent systems}.
\newblock \bibinfo{journal}{\emph{International Journal of Systems Science}}
  \bibinfo{volume}{47}, \bibinfo{number}{6} (\bibinfo{year}{2016}),
  \bibinfo{pages}{1366--1375}.
\newblock


\end{thebibliography}

\end{document}